\documentclass[11pt,a4paper]{article}

\usepackage[T1]{fontenc}
\usepackage{lmodern}
\usepackage{microtype}
\usepackage[a4paper,margin=3cm]{geometry}
\usepackage{amsmath,amssymb,amsthm,mathtools}
\usepackage{enumitem}
\usepackage{booktabs}
\usepackage{xcolor}
\usepackage{hyperref}
\hypersetup{colorlinks=true,linkcolor=blue,citecolor=blue,urlcolor=blue}
\usepackage[numbers,sort&compress]{natbib}

\theoremstyle{plain}
\newtheorem{theorem}{Theorem}[section]
\newtheorem{proposition}[theorem]{Proposition}
\newtheorem{lemma}[theorem]{Lemma}
\newtheorem{corollary}[theorem]{Corollary}

\theoremstyle{definition}
\newtheorem{definition}[theorem]{Definition}
\newtheorem{assumption}[theorem]{Assumption}

\theoremstyle{remark}
\newtheorem{remark}[theorem]{Remark}

\newcommand{\R}{\mathbb{R}}

\newcommand{\eps}{\varepsilon}

\title{Sorting and Global Uniqueness in Two-Good HARA Economies with Many Patience Types}

\author{Andrea Loi\thanks{Dipartimento di Matematica e Informatica, Universit\`a di Cagliari, Cagliari, Italy. Email: \texttt{loi@unica.it}.}
\and Stefano Matta\thanks{Dipartimento di Scienze Economiche e Aziendali, Universit\`a di Cagliari, Cagliari, Italy. Email: \texttt{smatta@unica.it}.}}

\date{\today}

\begin{document}
\maketitle

\begin{abstract}
We study global uniqueness of competitive equilibrium in two-good pure-exchange
economies with heterogeneous impatience types and a common HARA Bernoulli utility.
The paper connects the CRRA sorting result of \citet{GeanakoplosWalsh2018} 
with the line of HARA uniqueness results developed in
\citet{LoiMatta2022,LoiMatta2024}.
In the CRRA case, ordered endowments provide a sorting mechanism for uniqueness.
In the HARA case, uniqueness is known to hold for arbitrary endowments under the
curvature bound $\gamma\le I/(I-1)$, where $I$ is the number of impatience types.
For two types, the curvature restriction can be removed under a monotone sorting
condition linking patience and endowment composition. The present paper shows
that this high-curvature HARA sorting mechanism is not specific to the two-type
case.

Our main result proves global uniqueness for any finite number of impatience
types and any $\gamma>1$. If types can be ordered so that more patient agents
hold weakly more of the first good and weakly less of the second, then the
equilibrium price is globally unique. Thus the paper extends the two-type
high-curvature HARA result to a genuinely multi-type setting and complements the
arbitrary-endowment low-curvature result by replacing the low-curvature
restriction with an economically interpretable sorting restriction.

In the CRRA subcase ($b=0$), the ordered-endowment condition coincides with that
of \citet{GeanakoplosWalsh2018}, and our corollary recovers their uniqueness
result. The contribution of the present paper is therefore not the sorting
condition itself but its reach: the same ordered heterogeneity in patience and
endowment composition rules out multiplicity throughout the shifted HARA case
($b>0$), for any finite number of types and any $\gamma>1$, through a global
coefficient-ratio argument.
\end{abstract}



\section{Introduction}\label{sec:intro}

Uniqueness of competitive equilibrium is central to comparative statics, welfare analysis, and the empirical use of general equilibrium models.  If the same primitives support several equilibrium prices, then a policy experiment, a redistribution of endowments, or a calibration exercise may have no single prediction: small changes in primitives can move the economy across distinct equilibrium branches.  This issue is particularly relevant in two-good exchange economies interpreted as two-date economies, where the equilibrium price is an intertemporal price and heterogeneous impatience is a natural modelling device.

The Sonnenschein--Mantel--Debreu results show that standard assumptions on individual preferences place surprisingly weak restrictions on aggregate excess demand, even when preferences are smooth, monotone, and concave \citep{Debreu1974,Mantel1974,Sonnenschein1973}. The relevant question is therefore not whether uniqueness follows in unrestricted economies, but which economically meaningful restrictions restore a single-crossing aggregate market-clearing equation.

The closest point of departure for our analysis is \citet{GeanakoplosWalsh2018}, who study the same two-good impatience-type environment under a general Bernoulli utility. They study two-good economies with many impatience types and a common Bernoulli utility.  With identical endowments, nonincreasing absolute risk aversion implies downward-sloping aggregate demand at equilibrium and hence uniqueness and stability.  In the CRRA case, they also allow a monotone form of endowment heterogeneity (their Proposition~5): more patient types hold relatively more of one good and less of the other.  This is exactly the sorting condition studied below, so the ordered-endowment mechanism is not new in the CRRA case.  Their discussion also identifies the central economic difficulty for the general case: outside special cases, uniqueness should require heterogeneity not to be too large or badly aligned, but it is difficult to express this requirement in closed form.

Other recent approaches impose structure at the level of individual demand; for instance, \citet{Won2023} derives uniqueness results for heterogeneous CRRA economies by characterizing demand shapes, and \citet{TodaWalsh2024} survey recent advances on uniqueness and multiplicity.  Won's analysis remains within the CRRA class, while \citet{TodaWalsh2024} survey recent advances on
uniqueness and multiplicity and place such contributions within the broader literature. 

Our route differs both in scope and in method: it treats the HARA specification with $b\ge 0$ and $\gamma>1$, including the non-homothetic case $b>0$, and exploits the coefficient structure generated by the transformed market-clearing equation.
 HARA utility is a tractable and widely used subclass of DARA preferences, containing CRRA as the case $b=0$.  In \citet{LoiMatta2022}, we considered a two-good HARA economy with $I$ impatience types and arbitrary endowments.  The main result there shows that equilibrium is globally unique whenever
\[
\gamma\le \frac{I}{I-1}.
\]
This result is strong because it imposes no restriction on endowments.  Its limitation is that the admissible curvature region shrinks quickly with the number of types: for $I=3$ it becomes $\gamma\le3/2$, and as $I$ grows the threshold converges to one.  Thus, in high-curvature applications with several types, the low-curvature result of Loi and Matta~\cite{LoiMatta2022} does not provide a uniqueness criterion.

\citet{LoiMatta2024} address the high-curvature region for the two-type case.  For $I=2$, they show that uniqueness can be recovered for arbitrary $\gamma>1$ under a sorting restriction linking patience and endowment composition.  In our notation, if $\sigma_1<\sigma_2$, $e_1\le e_2$, and $f_1\ge f_2$, then the key coefficient determinant has the right sign and equilibrium is globally unique.  The natural next question is whether the same sorting logic extends beyond two types.

We give a positive answer.  The paper proves that the high-curvature sorting
mechanism identified in the two-type HARA economy of \citet{LoiMatta2024} extends
to every finite number of types.  More precisely, for any $\gamma>1$, global
uniqueness holds within the class of HARA economies in which impatience and
endowment composition are ordered in opposite directions across goods.  The
qualification HARA is essential: in the CRRA subcase $b=0$ the ordered-endowment
condition is already that of \citet{GeanakoplosWalsh2018}, and our result recovers
theirs (Remark~\ref{rem:crra-benchmark}).  The genuinely new content is the reach of
the mechanism into the non-homothetic HARA region $b>0$, for arbitrary $I$, through a
single global argument.  In this sense the two-type high-curvature theorem is the
first instance of a general multi-type \emph{HARA} sorting theorem, not an isolated
low-dimensional fact.

\medskip
\noindent\textbf{Main contribution.}
Let $\eps=1/\gamma$ and $\sigma_i=\beta_i^\eps$, where $\beta_i$ is type $i$'s
impatience parameter.  Denote by $(e_i,f_i)$ the endowment vector of type $i$.
The main sorting condition requires that types can be ordered so that
\begin{equation}\label{eq:intro-sorting}
0<\sigma_1<\cdots<\sigma_I,
\qquad
e_1\le\cdots\le e_I,
\qquad
f_1\ge\cdots\ge f_I .
\end{equation}
In words, more patient types are endowed with weakly more of the first good and
weakly less of the second good.

The proof is based on a transformation of the market-clearing condition.  Setting
$r=p^{1-\eps}$, the equilibrium equation can be written as
\begin{equation}\label{eq:intro-intersection}
 r^\theta=\frac{A(r)}{B(r)},
 \qquad r>0,
 \qquad \theta=\frac{\eps}{1-\eps}>0,
\end{equation}
where $A$ and $B$ are polynomials with positive coefficients generated by the
HARA demand system.  We write them as
\[
A(r)=\nu_{I-1}+\nu_{I-2}r+\cdots+\nu_0 r^{I-1},
\qquad
B(r)=\mu_I+\mu_{I-1}r+\cdots+\mu_1 r^{I-1},
\]
with $\nu_j>0$ and $\mu_j>0$.  Since the left-hand side of
\eqref{eq:intro-intersection} is strictly increasing, uniqueness follows if the
ratio $A(r)/B(r)$ is globally nonincreasing.

A sufficient condition for this monotonicity is the coefficient-ratio chain
\begin{equation}\label{eq:intro-chain}
\frac{\nu_{I-1}}{\mu_I}\ge
\frac{\nu_{I-2}}{\mu_{I-1}}\ge\cdots\ge
\frac{\nu_0}{\mu_1}.
\end{equation}

Equivalently, the adjacent cross-product inequalities associated with the two
coefficient arrays have the sign required to make the ratio $A(r)/B(r)$
nonincreasing. We refer to this condition as a TP$_2$ ratio chain, using the term
only as a shorthand for the corresponding adjacent $2\times2$ minors of the
two-row coefficient array, in the spirit of total positivity
\citep{Karlin1968}.

The key step of the paper is to prove that the sorting condition
\eqref{eq:intro-sorting} implies the chain \eqref{eq:intro-chain}.  After the
HARA endowment shift, the coefficients of $A$ and $B$ can be represented as
weighted sums of elementary symmetric polynomials.  

The ordering in \eqref{eq:intro-sorting} then yields precisely the coefficient-ratio chain \eqref{eq:intro-chain}, by showing that the relevant adjacent cross-product inequalities hold.
Hence $A(r)/B(r)$ is globally nonincreasing, whereas
$r^\theta$ is strictly increasing, and the transformed market-clearing equation
has at most one solution.

We focus deliberately on globally sorted heterogeneity. Within this sorted HARA class, uniqueness does not require a small-dispersion assumption: endowment dispersion may be large, provided endowment composition is ordered with patience. The complementary case of non-sorted economies close to the equal-share benchmark concerns a different region of the primitive space and is left for future work.

\medskip
\noindent\textbf{Organization.}
Section~\ref{sec:model} introduces the economy and derives the transformed market-clearing equation.  Section~\ref{sec:ratio} proves the general coefficient ratio criterion.  Section~\ref{sec:sorting} proves the main sorting theorem and the global uniqueness corollary.  Section~\ref{sec:econ} discusses the economic interpretation of sorting. Section~\ref{sec:conclusion} concludes.

\section{The economy and the transformed market-clearing equation}\label{sec:model}

\subsection{Primitives}

There are two goods and $I\ge2$ impatience types.  Type $i$ is endowed with $(e_i,f_i)\in\R_+^2$ and has preferences
\begin{equation}\label{eq:Ui}
U_i(x_i,y_i)=u_H(x_i)+\beta_i u_H(y_i),
\qquad \beta_i>0,
\end{equation}
where
\begin{equation}\label{eq:HARA}
u_H(z)=\frac{\gamma}{1-\gamma}\left(b+\frac{a}{\gamma}z\right)^{1-\gamma},
\qquad a>0,
\qquad b\ge0,
\qquad \gamma>1.
\end{equation}
We exclude the logarithmic case $\gamma=1$.  Set
\begin{equation}\label{eq:eps-sigma}
\eps:=\frac{1}{\gamma}\in(0,1),
\qquad
\sigma_i:=\beta_i^{\eps},
\qquad
h:=\frac{b}{a\eps},
\end{equation}
and define shifted endowments
\begin{equation}\label{eq:shifted}
E_i:=e_i+h,
\qquad
F_i:=f_i+h.
\end{equation}
We assume throughout that aggregate endowments in both goods are positive and that the shifted endowments used in the ratio arguments are positive. 

Let $p>0$ be the price of good $x$ in units of good $y$.  Type $i$ solves
\[
\max_{x_i,y_i\ge0} u_H(x_i)+\beta_i u_H(y_i)
\quad\text{s.t.}\quad
p x_i+y_i=p e_i+f_i.
\]
Write
\[
r_x:=\sum_{i=1}^I e_i,
\qquad
r_y:=\sum_{i=1}^I f_i.
\]

\subsection{Individual demand}

\begin{lemma}[Demand]\label{lem:demand}
For each \(p>0\), type \(i\)'s demand for good \(x\) is 

\begin{equation}\label{eq:demand-x}
x_i(p)=\frac{b-bp^\eps\sigma_i+a\eps(p e_i+f_i)}{a\eps(p+\sigma_i p^\eps)}.
\end{equation}
\end{lemma}

\begin{proof}
The first-order condition is
\[
u_H'(x_i)=\beta_i p\,u_H'(y_i).
\]
Since
\[
u_H'(z)=a\left(b+\frac{a}{\gamma}z\right)^{-\gamma},
\]
we obtain
\[
b+\frac{a}{\gamma}y_i=\sigma_i p^\eps\left(b+\frac{a}{\gamma}x_i\right).
\]
Using $y_i=p e_i+f_i-px_i$ and $\eps=1/\gamma$ gives \eqref{eq:demand-x}.
\end{proof}

Aggregate excess demand for good $x$ is
\begin{equation}\label{eq:Z}
Z(p):=\sum_{i=1}^I x_i(p)-r_x.
\end{equation}
An equilibrium price is a solution $p^*>0$ of $Z(p^*)=0$.

\subsection{Elementary symmetric polynomials}

The coefficient representation below uses elementary symmetric polynomials in the patience vector \(\sigma\). 
This is the same algebraic device used in \citet{LoiMatta2022}; we recall the notation here to keep the argument self-contained.

For a vector $\sigma=(\sigma_1,\ldots,\sigma_I)$ and an integer
$t=1,\ldots,I$, let $S_t(\sigma)$ denote the elementary symmetric polynomial
of degree $t$, namely
\[
S_t(\sigma)
=
\sum_{1\le i_1<\cdots<i_t\le I}
\sigma_{i_1}\cdots\sigma_{i_t}.
\]
We use the conventions
\[
S_0(\sigma)=1,
\qquad
S_t(\sigma)=0 \quad\text{if } t<0 \text{ or } t>I.
\]
For each $i$, let $\sigma_{-i}$ be the vector obtained from $\sigma$ by deleting
its $i$th component.  Thus $S_t(\sigma_{-i})$ is the elementary symmetric
polynomial of degree $t$ computed on all components of $\sigma$ except
$\sigma_i$.  We shall repeatedly use the identity
\begin{equation}\label{eq:St-identity}
S_t(\sigma)-S_t(\sigma_{-i})
=
\sigma_i S_{t-1}(\sigma_{-i}).
\end{equation}

\subsection{Coefficient representation}

Define, for $t=1,\ldots,I$,
\begin{equation}\label{eq:mu-def}
\mu_t:=\sum_{i=1}^I E_i\sigma_i S_{t-1}(\sigma_{-i}),
\end{equation}
and for $t=0,\ldots,I-1$,
\begin{equation}\label{eq:nu-def}
\nu_t:=\sum_{i=1}^I F_i S_t(\sigma_{-i}).
\end{equation}
All these coefficients are positive under the standing assumptions.

\noindent Finally, introduce the transformed price
\begin{equation}\label{eq:rtheta}
r:=p^{1-\eps}>0,
\qquad
\theta:=\frac{\eps}{1-\eps}=\frac{1}{\gamma-1}>0.
\end{equation}

\begin{proposition}[Market clearing as a one-dimensional intersection]\label{prop:intersection}
The equilibrium condition $Z(p)=0$ is equivalent to
\begin{equation}\label{eq:intersection}
r^\theta=\frac{A(r)}{B(r)},
\qquad r>0,
\end{equation}
where
\begin{equation}\label{eq:A-def}
A(r):=\sum_{t=0}^{I-1}\nu_t r^{I-1-t}
=\nu_{I-1}+\nu_{I-2}r+\cdots+\nu_0 r^{I-1},
\end{equation}
and
\begin{equation}\label{eq:B-def}
B(r):=\sum_{t=1}^{I}\mu_t r^{I-t}
=\mu_I+\mu_{I-1}r+\cdots+\mu_1 r^{I-1}.
\end{equation}
\end{proposition}

\begin{proof}
Substitute \eqref{eq:demand-x} into $Z(p)=0$ and multiply by the positive common denominator
\[
a\eps\prod_{i=1}^I(p+\sigma_i p^\eps).
\]
Using $p+\sigma_i p^\eps=p^\eps(r+\sigma_i)$, where $r=p^{1-\eps}$,
and then dividing by the positive factor $p^{(I-1)\eps}$, the market-clearing
condition becomes
\[
\sum_{i=1}^I\left[b+a\eps f_i+p^\eps(a\eps e_i r-b\sigma_i)\right]
\prod_{j\ne i}(r+\sigma_j)
-a\eps r_x p^\eps\prod_{i=1}^I(r+\sigma_i)=0.
\]
The terms not multiplied by $p^\eps$ are
\[
a\eps\sum_{i=1}^I F_i\prod_{j\ne i}(r+\sigma_j)
=a\eps\sum_{t=0}^{I-1}\nu_t r^{I-1-t}=a\eps A(r).
\]
Moving the terms multiplied by $p^\eps$ to the right-hand side gives
\[
a\eps A(r)
=
p^\eps\left[
a\eps r_x\prod_{i=1}^I(r+\sigma_i)
-\sum_{i=1}^I(a\eps e_i r-b\sigma_i)\prod_{j\ne i}(r+\sigma_j)
\right].
\]
For $t=1,\ldots,I$, the coefficient of $r^{I-t}$ in the bracket on the
right-hand side is
\[
a\eps\left[
r_xS_t(\sigma)-\sum_{i=1}^I e_iS_t(\sigma_{-i})
+h\sum_{i=1}^I\sigma_iS_{t-1}(\sigma_{-i})
\right].
\]
By \eqref{eq:St-identity},
\[
r_xS_t(\sigma)-\sum_{i=1}^I e_iS_t(\sigma_{-i})
=\sum_{i=1}^I e_i\sigma_iS_{t-1}(\sigma_{-i}),
\]
so this coefficient is
\[
a\eps\sum_{i=1}^I E_i\sigma_iS_{t-1}(\sigma_{-i})=a\eps\mu_t.
\]
Hence
\[
a\eps A(r)=p^\eps a\eps B(r),
\]
or equivalently
\[
A(r)-p^\eps B(r)=0.
\]
Since $p^\eps=(p^{1-\eps})^{\eps/(1-\eps)}=r^\theta$, this is exactly \eqref{eq:intersection}.
\end{proof}

\begin{remark}
The polynomial proofs in \citet{LoiMatta2022} and \citet{LoiMatta2024} often approximate $\eps$ by rational numbers in order to apply root-counting arguments to polynomials in $p^{1/n}$.  The present monotonicity proof does not require this step.  The change of variable $r=p^{1-\eps}$ is valid for every real $\eps\in(0,1)$.
\end{remark}

\section{A coefficient ratio criterion for uniqueness}\label{sec:ratio}

The transformed market-clearing equation \eqref{eq:intersection} has a strictly increasing left-hand side.  Hence uniqueness follows if the right-hand side $A(r)/B(r)$ is nonincreasing.

\begin{definition}[Ratio chain]\label{def:ratio-chain}
The coefficient ratio chain is
\begin{equation}\label{eq:ratio-chain}
\frac{\nu_{I-1}}{\mu_I}\ge
\frac{\nu_{I-2}}{\mu_{I-1}}\ge\cdots\ge
\frac{\nu_0}{\mu_1}.
\end{equation}
Equivalently, the adjacent cross-product inequalities
\begin{equation}\label{eq:D-def}
D_t:=\mu_{t+1}\nu_{t+1}-\mu_{t+2}\nu_t\ge0,
\qquad t=0,\ldots,I-2,
\end{equation}
hold.
Thus $D_t\ge0$ is precisely the link
\[
\frac{\nu_{t+1}}{\mu_{t+2}}\ge \frac{\nu_t}{\mu_{t+1}}
\]
of the chain, written with indices increasing from $0$ to $I-2$.
\end{definition}

The inequalities $D_t\ge0$ are precisely the adjacent $2\times2$ minors of the two-row coefficient array formed by $(\mu_1,\ldots,\mu_I)$ and $(\nu_0,\ldots,\nu_{I-1})$, aligned as in \eqref{eq:D-def}. We use the term TP$_2$ only in this limited coefficient-array sense.

\begin{lemma}[Ratio chain implies monotonicity]\label{lem:ratio-monotone}
If \eqref{eq:ratio-chain} holds, then $A(r)/B(r)$ is nonincreasing on $(0,\infty)$.
\end{lemma}

\begin{proof}
Write
\[
A(r)=\sum_{k=0}^{I-1}a_kr^k,
\qquad
B(r)=\sum_{k=0}^{I-1}b_kr^k,
\]
where $a_k=\nu_{I-1-k}$ and $b_k=\mu_{I-k}$.  The chain \eqref{eq:ratio-chain} says
\[
\frac{a_0}{b_0}\ge\frac{a_1}{b_1}\ge\cdots\ge\frac{a_{I-1}}{b_{I-1}}.
\]
A direct computation gives
\begin{equation}\label{eq:derivative-general}
A'(r)B(r)-A(r)B'(r)
=\sum_{0\le k<\ell\le I-1}(\ell-k)(a_\ell b_k-a_kb_\ell)r^{k+\ell-1}.
\end{equation}
For $\ell>k$, the ratio chain implies $a_\ell/b_\ell\le a_k/b_k$, hence $a_\ell b_k-a_kb_\ell\le0$.  Therefore $A'B-AB'\le0$ and $(A/B)'\le0$.
\end{proof}

\begin{theorem}[Uniqueness under the ratio chain]\label{thm:ratio-unique}
If the coefficient ratio chain \eqref{eq:ratio-chain} holds, then the equilibrium price is globally unique.
\end{theorem}

\begin{proof}
By Proposition~\ref{prop:intersection}, equilibria are the solutions of
\[
r^\theta=A(r)/B(r).
\]
The left-hand side is strictly increasing from $0$ to $+\infty$ on $(0,\infty)$.  By Lemma~\ref{lem:ratio-monotone}, the right-hand side is nonincreasing and positive.  Hence the two sides can cross at most once.  They cross at least once because, as $r\downarrow0$, $r^\theta-A(r)/B(r)\to-\nu_{I-1}/\mu_I<0$, while as $r\to\infty$, $r^\theta-A(r)/B(r)\to+\infty$.  Thus there is exactly one equilibrium price.
\end{proof}

\section[Sorting implies the TP2 ratio chain]{Sorting implies the TP$_2$ ratio chain}\label{sec:sorting}

We now prove the main result: sorting of patience and endowment composition implies the ratio chain \eqref{eq:ratio-chain}.  The proof has two ingredients.  First, the symmetric-polynomial ratio $S_{t+1}/S_t$ is increasing in each argument.  Second, a weighted average of a decreasing sequence falls when the likelihood ratio of the weights shifts mass toward higher indices.

\begin{assumption}[Sorting]\label{ass:sorting}
Types are ordered so that
\begin{equation}\label{eq:sorting}
0<\sigma_1<\cdots<\sigma_I,
\qquad
E_1\le\cdots\le E_I,
\qquad
F_1\ge\cdots\ge F_I.
\end{equation}
Equivalently, since $h$ is common to all types,
\[
e_1\le\cdots\le e_I,
\qquad
f_1\ge\cdots\ge f_I.
\]
\end{assumption}

\begin{lemma}[Monotonicity of symmetric-polynomial ratios]\label{lem:rho-monotone}
Fix $t\in\{0,\ldots,I-2\}$ and define
\begin{equation}\label{eq:rho-def}
\rho_i^{(t)}:=\frac{S_{t+1}(\sigma_{-i})}{S_t(\sigma_{-i})},
\qquad i=1,\ldots,I.
\end{equation}
If $0<\sigma_1<\cdots<\sigma_I$, then
\begin{equation}\label{eq:rho-decreasing}
\rho_1^{(t)}\ge\rho_2^{(t)}\ge\cdots\ge\rho_I^{(t)}.
\end{equation}
The inequalities are strict when the order of the $\sigma_i$ is strict.
\end{lemma}

\begin{proof}
Let $x=(x_1,\ldots,x_{I-1})$ have positive components and define
\[
R_t(x):=\frac{S_{t+1}(x)}{S_t(x)}.
\]
We first show that $R_t$ is increasing in each coordinate.  Differentiating with respect to $x_j$ and writing $x_{-j}$ for the vector with $x_j$ deleted, we obtain
\[
\frac{\partial R_t}{\partial x_j}
=\frac{S_t(x)S_t(x_{-j})-S_{t+1}(x)S_{t-1}(x_{-j})}{S_t(x)^2}.
\]
Using
\[
S_t(x)=S_t(x_{-j})+x_jS_{t-1}(x_{-j}),
\qquad
S_{t+1}(x)=S_{t+1}(x_{-j})+x_jS_t(x_{-j}),
\]
the numerator reduces to
\[
S_t(x_{-j})^2-S_{t+1}(x_{-j})S_{t-1}(x_{-j}),
\]
which is nonnegative by Newton's inequalities for elementary symmetric
polynomials.  In fact, in the present range of indices and for positive
arguments, this quantity is strictly positive.  Hence $R_t$ is strictly
increasing in each coordinate.

If $i<j$, then the sorted vector $\sigma_{-i}$ componentwise dominates the
sorted vector $\sigma_{-j}$: deleting the smaller element $\sigma_i$ leaves a
vector no smaller in every coordinate than the vector obtained by deleting the
larger element $\sigma_j$.  Moreover, when
$0<\sigma_1<\cdots<\sigma_I$, this dominance is strict in at least one
coordinate.  Since $R_t$ is strictly increasing in each coordinate, we obtain
\[
R_t(\sigma_{-i})>R_t(\sigma_{-j}),
\]
and therefore
\[
\rho_i^{(t)}>\rho_j^{(t)}
\qquad\text{whenever } i<j.
\]
This proves \eqref{eq:rho-decreasing}, with strict inequalities under strict
ordering of the $\sigma_i$.
\end{proof}

\begin{lemma}[Weighted rearrangement inequality]\label{lem:rearrangement}
Let $(\rho_i)$ be nonincreasing and let $(\lambda_i)$ be nondecreasing.  For any positive weights $(b_i)$, define $a_i=\lambda_i b_i$.  Then
\begin{equation}\label{eq:rearrangement}
\frac{\sum_i a_i\rho_i}{\sum_i a_i}
\le
\frac{\sum_i b_i\rho_i}{\sum_i b_i}.
\end{equation}
\end{lemma}

\begin{proof}
Let $\bar\rho_b=\sum_i b_i\rho_i/\sum_i b_i$.  Inequality \eqref{eq:rearrangement} is equivalent to
\[
\sum_i b_i\lambda_i(\rho_i-\bar\rho_b)\le0.
\]
Using the standard identity for a weighted covariance,
\[
\sum_i b_i\lambda_i(\rho_i-\bar\rho_b)
=\frac{1}{2\sum_i b_i}
\sum_{i,j}b_ib_j(\lambda_i-\lambda_j)(\rho_i-\rho_j).
\]
Since $(\lambda_i)$ is nondecreasing and $(\rho_i)$ is nonincreasing, every product
$(\lambda_i-\lambda_j)(\rho_i-\rho_j)$ is nonpositive.  The result follows.
\end{proof}

\begin{theorem}[Sorting implies the TP$_2$ ratio chain]\label{thm:sorting-ratio}
Under Assumption~\ref{ass:sorting}, for every $t=0,\ldots,I-2$,
\begin{equation}\label{eq:D-positive-sorting}
D_t=\mu_{t+1}\nu_{t+1}-\mu_{t+2}\nu_t\ge0.
\end{equation}
Equivalently, the coefficient ratio chain \eqref{eq:ratio-chain} holds.
\end{theorem}

\begin{proof}
Fix $t\in\{0,\ldots,I-2\}$ and set
\[
\rho_i:=\rho_i^{(t)}=\frac{S_{t+1}(\sigma_{-i})}{S_t(\sigma_{-i})}.
\]
By Lemma~\ref{lem:rho-monotone}, $(\rho_i)$ is nonincreasing in $i$.

Now write
\[
b_i:=F_iS_t(\sigma_{-i}),
\qquad
a_i:=E_i\sigma_iS_t(\sigma_{-i}).
\]
Then
\[
\nu_t=\sum_i b_i,
\qquad
\nu_{t+1}=\sum_i b_i\rho_i,
\]
and
\[
\mu_{t+1}=\sum_i a_i,
\qquad
\mu_{t+2}=\sum_i a_i\rho_i.
\]
Moreover,
\[
\frac{a_i}{b_i}=\frac{E_i\sigma_i}{F_i}.
\]
Under Assumption~\ref{ass:sorting}, $E_i$ and $\sigma_i$ are nondecreasing, while $F_i$ is nonincreasing and positive.  Therefore $a_i/b_i$ is nondecreasing in $i$.

Applying Lemma~\ref{lem:rearrangement} gives
\[
\frac{\mu_{t+2}}{\mu_{t+1}}
=\frac{\sum_i a_i\rho_i}{\sum_i a_i}
\le
\frac{\sum_i b_i\rho_i}{\sum_i b_i}
=\frac{\nu_{t+1}}{\nu_t}.
\]
Multiplying by the positive denominator $\mu_{t+1}\nu_t$ yields
\[
\mu_{t+1}\nu_{t+1}-\mu_{t+2}\nu_t\ge0.
\]
This proves \eqref{eq:D-positive-sorting} for every $t$ and hence the ratio chain.
\end{proof}

\begin{corollary}[Global uniqueness under sorting]\label{cor:sorting-unique}
For any finite $I\ge2$ and any $\gamma>1$, if Assumption~\ref{ass:sorting} holds, then the competitive equilibrium price is globally unique.
\end{corollary}

\begin{proof}
By Theorem~\ref{thm:sorting-ratio}, sorting implies the coefficient ratio chain.  By Theorem~\ref{thm:ratio-unique}, the ratio chain implies global uniqueness.
\end{proof}

\begin{remark}[The CRRA benchmark and the HARA extension]\label{rem:crra-benchmark}
When $b=0$, HARA utility reduces to CRRA and Assumption~\ref{ass:sorting} is exactly the ordered-endowment restriction of \citet[Proposition~5]{GeanakoplosWalsh2018}: more patient types hold weakly more of the first good and weakly less of the second.  In that subcase Corollary~\ref{cor:sorting-unique} recovers their CRRA uniqueness result, an equivalence already recorded in \citet[Remark~4]{LoiMatta2022} for the two-type ($\gamma=3$) case.  The new content of the present paper is therefore not the sorting condition itself but its reach: the same ordered-endowment mechanism delivers global uniqueness throughout the non-homothetic HARA class ($b>0$), for every finite $I$ and every $\gamma>1$.

The argument is also different in kind.  \citet{GeanakoplosWalsh2018} obtain downward-sloping aggregate demand local to equilibrium through an income-effect (covariance) argument; here market clearing is transformed globally into $r^\theta=A(r)/B(r)$, and global monotonicity of $A(r)/B(r)$ follows from the coefficient-ratio chain.  
The two routes agree in their conclusion when $b=0$. For $b>0$, the income-effect ordering that drives the CRRA proof is no longer directly available (cf.\ \citealp{GeanakoplosWalsh2018}), and the present coefficient-ratio argument supplies the extension beyond the homothetic benchmark.
\end{remark}

\begin{remark}
The ordering of endowments in Assumption~\ref{ass:sorting} may be weak.  Since $\sigma_i$ is strictly increasing and $E_i,F_i>0$, the likelihood ratio
\[
\frac{E_i\sigma_i}{F_i}
\]
is strictly increasing even when all $E_i$ are equal and all $F_i$ are equal.  Moreover, Lemma~\ref{lem:rho-monotone} gives a strictly decreasing sequence $\rho_i^{(t)}$ when the $\sigma_i$ are strictly ordered.  Hence the adjacent cross-product inequalities in Theorem~\ref{thm:sorting-ratio} are strict under the maintained nondegeneracy conditions.
\end{remark}

\begin{remark}[Operational test]\label{rem:operational-test}
The sorting theorem gives a direct test for uniqueness.  First compute $\sigma_i=\beta_i^{1/\gamma}$ and order types so that $\sigma_1<\cdots<\sigma_I$.  Then check whether endowments can be ordered as $e_1\le\cdots\le e_I$ and $f_1\ge\cdots\ge f_I$.  If so, Corollary~\ref{cor:sorting-unique} gives global uniqueness for every $\gamma>1$.  If this sorting test fails, one can still compute the coefficients $(\mu_t)$ and $(\nu_t)$ and verify the ratio chain directly.
\end{remark}

\subsection{A numerical illustration}\label{subsec:numerical-illustration}

The criterion is easy to evaluate, and we choose a genuinely non-homothetic instance.  Consider $I=3$, $\gamma=4$, $a=1$, $b=1$, and
\[
\sigma=(1,2,4),\qquad e=(1,3,10),\qquad f=(10,3,1).
\]
This economy is sorted and substantially polarized.  Since $\eps=1/\gamma=1/4$, the endowment shift is $h=b/(a\eps)=4$, so
\[
E=e+h=(5,7,14),\qquad F=f+h=(14,7,5).
\]
From \eqref{eq:mu-def}--\eqref{eq:nu-def},
\[
\mu_1=75,
\qquad
\mu_2=268,
\qquad
\mu_3=208,
\]
and
\[
\nu_0=26,
\qquad
\nu_1=134,
\qquad
\nu_2=150.
\]
Therefore
\[
\frac{\nu_2}{\mu_3}=\frac{150}{208}\simeq0.721,
\qquad
\frac{\nu_1}{\mu_2}=\frac{134}{268}=0.5,
\qquad
\frac{\nu_0}{\mu_1}=\frac{26}{75}\simeq0.347.
\]
The ratio chain is strict.  Equivalently, the adjacent cross-product terms are
\[
D_0=\mu_1\nu_1-\mu_2\nu_0=3082>0,
\qquad
D_1=\mu_2\nu_2-\mu_3\nu_1=12328>0.
\]
Thus the transformed market-clearing ratio is decreasing, and the equilibrium price is unique.  Since $b>0$, this example lies strictly outside the CRRA case covered by \citet{GeanakoplosWalsh2018}.

\begin{remark}[Relation with the two-type theorem]
When $I=2$, Theorem~\ref{thm:sorting-ratio} reduces to the determinant inequality used in \citet{LoiMatta2024}.  Thus, the two-type HARA result is the first instance of a general multi-type HARA monotone-likelihood-ratio mechanism.  The present theorem shows that, within the sorted HARA class, no additional ``small dispersion'' restriction is needed, even for many types.
\end{remark}

\begin{remark}[Beyond sorted economies]\label{rem:beyond-sorted}
The present paper is intentionally restricted to the global sorting mechanism.  The sorting condition is not a smallness assumption: endowments may be highly dispersed as long as they are ordered across patience types.  A different and complementary problem is to guarantee uniqueness when endowments are not ordered.  That local question can be addressed by studying neighborhoods of the equal-share benchmark, where the same TP\(_2\) ratio chain follows from Newton inequalities and continuity.  We leave that non-sorted equal-share approach for future work, so that the global sorting theorem remains conceptually distinct.
\end{remark}

\section{Economic interpretation of sorting}\label{sec:econ}

The sorting condition
\begin{equation*}
0<\sigma_1<\cdots<\sigma_I,
\qquad e_1\le\cdots\le e_I,
\qquad f_1\ge\cdots\ge f_I
\end{equation*}
is a single-crossing restriction. More patient types are relatively more exposed to one good and less exposed to the other. In a two-date interpretation, where the two goods are dated commodities, the condition says that impatience and intertemporal endowment composition are ordered in a coherent way.

The proof shows why this ordering matters. For each adjacent coefficient inequality $D_t\ge0$, the relevant object is a comparison between two averages of
\begin{equation*}
\rho_i^{(t)}=\frac{S_{t+1}(\sigma_{-i})}{S_t(\sigma_{-i})}.
\end{equation*}
The elementary symmetric polynomial $S_t(\sigma_{-i})$ summarizes the $t$th-order aggregate composition of patience among all types except $i$. Thus the ratio $\rho_i^{(t)}$ measures the marginal increase in this residual patience composition when one moves from products of order $t$ to products of order $t+1$. If a low-index, less patient type is removed, the remaining vector $\sigma_{-i}$ is relatively more patient; if a high-index, more patient type is removed, the remaining vector is relatively less patient. Hence the values $\rho_i^{(t)}$ are high for low-index types and low for high-index types.

Sorting makes the $\mu$-weights relatively larger on high-index types than the $\nu$-weights. Hence the $\mu$-weighted average of $\rho_i^{(t)}$ is lower than the $\nu$-weighted average. This is exactly the sign condition needed for monotonicity of aggregate market clearing.

Thus, within the HARA coefficient representation, sorting is not merely an interpretation: it is the primitive condition that makes the transformed market-clearing ratio monotone, and hence the equilibrium price unique, throughout the high-curvature region.

\section{Concluding remarks}\label{sec:conclusion}

This paper proves global uniqueness in multi-type two-good HARA exchange economies under a sorted-endowment condition. The result holds for every finite number of types and every \(\gamma>1\). In the CRRA subcase \(b=0\), it recovers the ordered-endowment benchmark of \citet[Proposition~5]{GeanakoplosWalsh2018}. Within the HARA programme, it extends the two-type high-curvature result of \citet{LoiMatta2024} to an arbitrary finite number of types and complements the arbitrary-endowment low-curvature result of \citet{LoiMatta2022}.

The key step is the unified coefficient representation
\[
\mu_t=\sum_i E_i\sigma_i S_{t-1}(\sigma_{-i}),
\qquad
\nu_t=\sum_i F_i S_t(\sigma_{-i}).
\]
Under sorting, the likelihood ratio \(E_i\sigma_i/F_i\) is increasing, while the elementary-symmetric ratio \(S_{t+1}(\sigma_{-i})/S_t(\sigma_{-i})\) is decreasing. A rearrangement argument gives the TP\(_2\) coefficient-ratio chain, which makes the market-clearing ratio nonincreasing. Since the transformed price term is strictly increasing, the equilibrium price is unique.

The paper deliberately isolates this global sorting mechanism from the complementary local question of non-sorted economies near equal shares. Several extensions remain natural. First, the ratio-chain test can be used as a computational determinacy diagnostic in calibrated heterogeneous-agent environments, including cases where sorting fails but the coefficient inequalities can be checked directly. Second, the total-positivity structure uncovered here may be useful beyond HARA economies, especially in models where aggregate market clearing can be reduced to ratios of coefficient arrays generated by elementary symmetric polynomials.

\section*{Acknowledgements}

Stefano Matta gratefully acknowledges financial support from Fondazione di Sardegna, grant no. F83C26000350007.


\end{document}